\documentclass[runningheads,envcountsame,a4paper]{llncs}

\usepackage{amsmath,enumerate,extarrows,mathtools,microtype,multicol}
\usepackage{thmtools,thm-restate}

\usepackage{amssymb}
\usepackage{graphicx}

\usepackage{url}
\urldef{\mailsa}\path|smithtim@ccs.neu.edu|

\begin{document}

\mainmatter  

\title{On Infinite Words Determined by Indexed Languages\thanks{This is the full version of a paper accepted for publication at MFCS 2014.  It contains an appendix with proofs which were omitted or only sketched in the body.}}





%
%
\author{Tim Smith}
%

\institute{Northeastern University\\
Boston, MA, USA\\
\mailsa}

%
%

\maketitle

\begin{abstract}

We characterize the infinite words determined by indexed languages.  An infinite language $L$ determines an infinite word $\alpha$ if every string in $L$ is a prefix of $\alpha$.  If $L$ is regular or context-free, it is known that $\alpha$ must be ultimately periodic.  We show that if $L$ is an indexed language, then $\alpha$ is a morphic word, i.e., $\alpha$ can be generated by iterating a morphism under a coding.  Since the other direction, that every morphic word is determined by some indexed language, also holds, this implies that the infinite words determined by indexed languages are exactly the morphic words.  To obtain this result, we prove a new pumping lemma for the indexed languages, which may be of independent interest.

\end{abstract}

\section{Introduction}

Formal languages and infinite words can be related to each other in various ways.  One natural connection is via the notion of a prefix language.  A prefix language is a language $L$ such that for all $x,y \in L$, $x$ is a prefix of $y$ or $y$ is a prefix of $x$.  Every infinite prefix language determines an infinite word.  Prefix languages were introduced by Book \cite{book} in an attempt to study the complexity of infinite words in terms of acceptance and generation by automata.  Book used the pumping lemma for context-free languages to show that every context-free prefix language is regular, implying that any infinite word determined by such a language is ultimately periodic.  Recent work has continued and expanded Book's project, classifying the infinite words determined by various classes of automata \cite{smith3} and parallel rewriting systems \cite{smith2}.

In this paper we characterize the infinite words determined by indexed languages.  The indexed languages, introduced in 1968 by Alfred Aho \cite{aho}, fall between the context-free and context-sensitive languages in the Chomsky hierarchy.  More powerful than the former class and more tractable than the latter, the indexed languages have been applied to the study of natural languages \cite{gazdar} in computational linguistics.  Indexed languages are generated by indexed grammars, in which nonterminals are augmented with stacks which can be pushed, popped, and copied to other nonterminals as the derivation proceeds.  Two automaton characterizations are the nested stack automata of \cite{aho2}, and the order-2 pushdown automata within the Maslov pushdown hierarchy \cite{maslov}.

The class of indexed languages IL includes all of the stack automata classes whose infinite words are characterized in \cite{smith3}, as well as all of the rewriting system classes whose infinite words are characterized in \cite{smith2}.  In particular, IL properly includes ET0L \cite{ers}, a broad class within the hierarchy of parallel rewriting systems known as L systems.  L systems have close connections with a class of infinite words called morphic words, which are generated by repeated application of a morphism to an initial symbol, under a coding \cite{as}.  In \cite{smith2} it is shown that every infinite word determined by an ET0L language is morphic.  This raises the question of whether the indexed languages too determine only morphic words, or whether indexed languages can determine infinite words which are not morphic.

To answer this question, we employ a new pumping lemma for IL.  In Book's paper, as well as in \cite{smith2} and \cite{smith3}, pumping lemmas played a prominent role in characterizing the infinite words determined by various language classes.  A pumping lemma for a language class $C$ is a powerful tool for proving that certain languages do not belong to $C$, and thereby for proving that certain infinite words cannot be determined by any language in $C$.  For the indexed languages, a pumping lemma exists due to Hayashi \cite{hayashi}, as well as a ``shrinking lemma'' due to Gilman \cite{gilman}.  We were not successful in using these lemmas to characterize the infinite words determined by IL, so instead have proved a new pumping lemma for this class (Theorem \ref{pumpable}), which may be of independent interest.

Our lemma generalizes a pumping lemma recently proved for ET0L languages \cite{rabkin}.  Roughly, it states that for any indexed language $L$, any sufficiently long word $w \in L$ may be written as $u_1 \dotsm u_n$, each $u_i$ may be written as $v_{i,1} \dotsm v_{i,n_i}$, and the $v_{i,j}$s may be replaced with $u_i$s to obtain new words in $L$.  Using this lemma, we extend to IL a theorem about frequent and rare symbols proved in \cite{rabkin} for ET0L, which can be used to prove that certain languages are not indexed.  We also use the lemma to obtain the new result that every infinite indexed language has an infinite subset in a smaller class of L systems called CD0L.  This implies that every infinite word determined in IL can also be determined in CD0L, and thus that every such word is morphic.  Since every morphic word can be determined by some CD0L language \cite{smith2}, we therefore obtain a complete characterization of the infinite words determined by indexed languages: they are exactly the morphic words.



\subsection{Proof techniques}

Our pumping lemma for IL generalizes the one proved in \cite{rabkin} for ET0L.  Derivations in an ET0L system, like those in an indexed grammar, can be viewed as having a tree structure, but with certain differences.  In ET0L, symbols are rewritten in parallel, and the tree is organized into levels corresponding to the steps of the derivation.  Further, each node in the tree has one of a finite set of possible labels, corresponding to the symbols in the ET0L system.  The proof in \cite{rabkin} classifies each level of the tree according to the set of symbols which appear at that level, and then finds two levels with the same symbol set, which are used to construct the pumping operation.  By contrast, the derivation tree of an indexed grammar is not organized into levels in this way, and there is no bound on the number of possible labels for the nodes, since each nonterminal can have an arbitrarily large stack.  We deal with these differences by assigning each node a ``type'' based on the set of nonterminals which appear among its descendants immediately before its stack is popped.  These types then play a role analogous to the symbol sets of \cite{rabkin} in our construction of the pumping operation.


%


\subsection{Related work}

The model used in this paper, in which infinite words are determined by languages of their prefixes, originates in Book's 1977 paper \cite{book}.  Book formulated the ``prefix property'' in order to allow languages to ``approximate'' infinite sequences, and showed that for certain classes of languages, if a language in the class has the prefix property, then it is regular.  A follow-up by Latteux \cite{latteux} gives a necessary and sufficient condition for a prefix language to be regular.  Languages whose complement is a prefix language, called ``coprefix languages'', have also been studied; see Berstel \cite{berstel} for a survey of results on infinite words whose coprefix language is context-free.  In Smith \cite{smith2}, prefix languages are used to categorize the infinite words determined by a hierarchy of L system classes.  In Smith \cite{smith3}, they are used to characterize the infinite words determined by several classes of one-way stack automata, and also studied in connection with multihead deterministic finite automata.

Hayashi's 1973 pumping lemma for indexed languages is proved in a dense thirty-page paper \cite{hayashi}. The main theorem states that if a given terminal derivation tree is big enough, new terminal derivation trees can be generated by the insertion of other trees into the given one.  Hayashi applies his theorem to give a new proof that the finiteness problem for indexed languages is solvable and to show that certain languages are not indexed.  Gilman's 1996 ``shrinking lemma'' for indexed languages \cite{gilman} is intended to be easier to employ, and operates directly on terminal strings rather than on derivation trees.  Our lemma generalizes the recent ET0L pumping lemma of Rabkin \cite{rabkin}.  Like Gilman's lemma, it is stated in terms of strings rather than derivation trees, making it easier to employ, while like Hayashi's lemma and unlike Gilman's, it provides a pumping operation which yields an infinity of new strings in the language.

Another connection between indexed languages and morphic words comes from Braud and Carayol\cite{bc}, in which morphic words are related to a class of graphs at level 2 of the pushdown hierarchy.  The string languages at this level of the hierarchy are the indexed languages.


\subsection{Outline of paper}

The paper is organized as follows.  Section \ref{sec:preliminaries} gives preliminary definitions and propositions.  Section \ref{sec:pumping} gives our pumping lemma for indexed languages.  Section \ref{sec:applications} gives applications for the lemma, in particular characterizing the infinite words determined by indexed languages.  Section \ref{sec:conclusion} gives our conclusions.

\section{Preliminaries}\label{sec:preliminaries}

An \textbf{alphabet} $A$ is a finite set of symbols.  A \textbf{word} is a concatenation of symbols from $A$.  We denote the set of finite words by $A^*$ and the set of infinite words by $A^\omega$.  A \textbf{string} $x$ is an element of $A^*$.  The length of $x$ is denoted by $|x|$.  We denote the empty string by \textbf{$\lambda$}.  For a symbol $c$, $\#_c(x)$ denotes the number of appearances of $c$ in $x$, and for an alphabet $B$, $\#_B(x)$ denotes $\sum_{c \in B} \#_c(x)$.  A \textbf{language} is a subset of $A^*$.  A (symbolic) \textbf{sequence} $S$ is an element of $A^* \cup A^\omega$.  A \textbf{prefix} of $S$ is a string $x$ such that $S = x S'$ for some sequence $S'$.  A \textbf{subword} (or factor) of $S$ is a string $x$ such that $S = w x S'$ for some string $w$ and sequence $S'$.  For $i \geq 1$, $S[i]$ denotes the $i$th symbol of $S$.  For a string $x \neq \lambda$, $x^\omega$ denotes the infinite word $xxx\dotsm$.  An infinite word of the form $xy^\omega$, where $x$ and $y$ are strings and $y \neq \lambda$, is called \textbf{ultimately periodic}.

\subsection{Prefix languages}

A \textbf{prefix language} is a language $L$ such that for all $x,y \in L$, $x$ is a prefix of $y$ or $y$ is a prefix of $x$.  A language $L$ \textbf{determines} an infinite word $\alpha$ iff $L$ is infinite and every $x \in L$ is a prefix of $\alpha$.  For example, the infinite prefix language \{$\lambda$, \texttt{ab}, \texttt{abab}, \texttt{ababab}, $\dotsc$\} determines the infinite word $(\texttt{ab})^\omega$.  For a language class $C$, let $\omega(C)$ = $\{\alpha \mid$ $\alpha$ is an infinite word determined by some $L \in C\}$.  The following propositions are basic consequences of the definitions.

\begin{remark}A language determines at most one infinite word.\end{remark}

\begin{remark}A language $L$ determines an infinite word iff $L$ is an infinite prefix language.\end{remark}

\begin{remark}If a language $L$ determines an infinite word $\alpha$ and $L'$ is an infinite subset of $L$, then $L'$ determines $\alpha$.\end{remark}

\subsection{Morphic words}

A \textbf{morphism} on an alphabet $A$ is a map $h$ from $A^*$ to $A^*$ such that for all $x,y \in A^*$, $h(xy) = h(x)h(y)$.  Notice that $h(\lambda) = \lambda$.  The morphism $h$ is a \textbf{coding} if for all $a \in A$, $|h(a)| = 1$.  A string $x \in A^*$ is \textbf{mortal} (for $h$) if there is an $m \geq 0$ such that $h^m(x) = \lambda$.  The morphism $h$ is \textbf{prolongable} on a symbol $a$ if $h(a) = ax$ for some $x \in A^*$, and $x$ is not mortal.  If $h$ is prolongable on $a, h^\omega(a)$ denotes the infinite word $a\ x\ h(x)\ h^2(x)\ \dotsm$.  An infinite word $\alpha$ is \textbf{morphic} if there is a morphism $h$, coding $e$, and symbol $a$ such that $h$ is prolongable on $a$ and $\alpha = e(h^\omega(a))$.  For example, let:
\[
 \begin{matrix*}[l]
  h(\texttt{s}) = \texttt{sbaa} & \ e(\texttt{s}) = \texttt{a} \\
  h(\texttt{a}) = \texttt{aa} & \ e(\texttt{a}) = \texttt{a} \\
  h(\texttt{b}) = \texttt{b} & \ e(\texttt{b}) = \texttt{b}
 \end{matrix*}
\]
Then $e(h^\omega(\texttt{s})) = \texttt{a}^1\texttt{ba}^2\texttt{ba}^4\texttt{ba}^8\texttt{ba}^{16}\texttt{b}\dotsm$ is a morphic word.  See \cite{as} for more on morphic words.  Morphic words have close connections with the parallel rewriting systems known as L systems.  Many classes of L systems appear in the literature; here we define only HD0L and CD0L.  For more on L systems, including the class ET0L, see \cite{krs} and \cite{rs}.  An \textbf{HD0L system} is a tuple $G = (A,h,w,g)$ where $A$ is an alphabet, $h$ and $g$ are morphisms on $A$, and $w$ is in $A^*$.  The language of $G$ is $L(G) = \{g(h^i(w)) \mid i \geq 0\}$.  If $g$ is a coding, $G$ is a \textbf{CD0L system}.  \textbf{HD0L} and \textbf{CD0L} are the sets of HD0L and CD0L languages, respectively.  From \cite{krs} and \cite{ers} we have CD0L $\subset$ HD0L $\subset$ ET0L $\subset$ IL.  In \cite{smith2} it is shown that $\omega$(CD0L) = $\omega$(HD0L) = $\omega$(ET0L), and $\alpha$ is in this class of infinite words iff $\alpha$ is morphic.

\subsection{Indexed languages}

The class of indexed languages IL consists of the languages generated by indexed grammars.  These grammars extend context-free grammars by giving each nonterminal its own stack of symbols, which can be pushed, popped, and copied to other nonterminals as the derivation proceeds.  Indexed grammars come in several forms \cite{aho,gazdar,hu}, all generating the same class of languages, but varying with respect to notation and which productions are allowed.  The following definition follows the form of \cite{hu}. 

An \textbf{indexed grammar} is a tuple $G = (N,T,F,P,S)$ in which $N$ is the nonterminal alphabet, $T$ is the terminal alphabet, $F$ is the stack alphabet, $S \in N$ is the start symbol, and $P$ is the set of productions of the forms




\begin{center}
  $A \rightarrow r$ \hspace{1cm} $A \rightarrow Bf$ \hspace{1cm} $Af \rightarrow r$
\end{center}

with $A,B \in N$, $f \in F$, and $r \in (N \cup T)^*$.  In an expression of the form $A f_1 \dotsm f_n$ with $A \in N$ and $f_1, \dotsc, f_n \in F$, the string $f_1 \dotsm f_n$ can be viewed as a stack joined to the nonterminal $A$, with $f_1$ denoting the top of the stack and $f_n$ the bottom.  For $r \in (N \cup T)^*$ and $x \in F^*$, we write $r\{x\}$ to denote $r$ with every $A \in N$ replaced by $Ax$.  For example, with $A,B \in N$ and $\texttt{c},\texttt{d} \in T$, $\texttt{cd}AB\{f\} = \texttt{cd}AfBf$.  For $q,r \in (N F^* \cup T)^*$, we write $q \longrightarrow r$ if there are $q_1,q_2 \in (N F^* \cup T)^*$, $A \in N$, $p \in (N \cup T)^*$, and $x,y \in F^*$ such that $q = q_1\ A x\ q_2$, $r = q_1\ p\{y\}\ q_2$, and one of the following is true: (1) $A \rightarrow p$ is in $P$ and $y = x$, (2) $A \rightarrow pf$ is in $P$ and $y = fx$, or (3) $Af \rightarrow p$ is in $P$ and $x = fy$.  Let $\xlongrightarrow{*}$ be the reflexive, transitive closure of $\longrightarrow$.  For $A \in N$ and $x \in F^*$, let $L(Ax) = \{s \in T^* \mid Ax \xlongrightarrow{*} s\}$.  The language of $G$, denoted $L(G)$, is $L(S)$.  The class IL of indexed languages is \{$L(G) \mid G$ is an indexed grammar\}.











See Example \ref{example} and Figure \ref{fig:tree} for a sample indexed grammar and derivation tree.  For convenience, we will work with a form of indexed grammar which we call ``grounded'', in which terminal strings are produced only at the bottom of the stack.  $G$ is \textbf{grounded} if there is a symbol $\$ \in F$ (called the bottom-of-stack symbol) such that every production has one of the forms

\begin{center}
  $S \rightarrow A\$$ \hspace{1cm} $A \rightarrow r$ \hspace{1cm} $A \rightarrow Bf$ \hspace{1cm} $Af \rightarrow r$ \hspace{1cm} $A\$ \rightarrow s$
\end{center}

with $A,B \in N \setminus S$, $f \in F \setminus \$$, $r \in (N \setminus S)^+$, and $s \in T^*$.  It is not difficult to verify the following proposition.

\begin{restatable}{proposition}{grounded}
  \label{grounded}
For every indexed grammar $G$, there is a grounded indexed grammar $G'$ such that $L(G') = L(G)$.
\end{restatable}

\section{Pumping Lemma for Indexed Languages}\label{sec:pumping}

In this section we present our pumping lemma for indexed languages (Theorem~\ref{pumpable}) and give an example of its use.  Our pumping lemma generalizes the ET0L pumping lemma of \cite{rabkin}.  Like that lemma, it allows positions in a word to be designated as ``marked'', and then provides guarantees about the marked positions during the pumping operation.





%
%

Let $G = (N,T,F,P,S)$ be a grounded indexed grammar.  To prove Theorem \ref{pumpable}, we will first prove a lemma about paths in derivation trees of $G$.  A derivation tree $D$ of a string $s$ has the following structure.  Each internal node of $D$ has a label in $N F^*$ (a nonterminal with a stack), and each leaf has a label in $T^*$ (a terminal string).  Each internal node has either a single leaf node as a child, or one or more internal children.  The root of $D$ is labelled by the start symbol $S$, and the terminal yield of $D$ is the string $s$.

If the string $s$ contains marked positions, then we will take $D$ to be marked in the following way.  Mark every leaf whose label contains a marked position of $s$, and then mark every internal node which has a marked descendant.  Call any node with more than one marked child a \textbf{branch node}.

A \textbf{path} $H$ in $D$ is a list of nodes $(v_0, \dotsc, v_m)$ with $m \geq 0$ such that for each $1 \leq i \leq m$, $v_i$ is a child of $v_{i-1}$.  For convenience, we will sometimes refer to nodes in $H$ by their indices; e.g. node $i$ in the context of $H$ means $v_i$.  When we say that there is a branch node between $i$ and $j$ we mean that the branch node is between $v_i$ (inclusive) and $v_j$ (exclusive).

We define several operations on internal nodes of $D$.  Each such node $v$ has the label $Ax$ for some $A \in N$ and $x \in F^*$.  Let $\sigma(v) = A$ and $\eta(v) = |x|$.  $\sigma(v)$ gives the nonterminal symbol of $v$ and $\eta(v)$ gives the height of $v$'s stack.  We say that a node $v'$ is in the scope of $v$ iff $v'$ is an internal node and there is a path in $D$ from $v$ to $v'$ such that for every node $v''$ on the path (including $v'$), $\eta(v'') \geq \eta(v)$.  Let $\beta(v)$ be the set of nodes $v'$ such that $v'$ is in the scope of $v$ but no child of $v'$ is in the scope of $v$.  The set $\beta(v)$ can be viewed as the ``last'' nodes in the scope of $v$.  Notice that for all $v' \in \beta(v)$, $\eta(v') = \eta(v)$.  Finally, we give $v$ a ``type'' $\tau(v)$ based on which nonterminal symbols appear in $\beta(v)$.  Let $\tau(v)$ be a 3-tuple such that:

\begin{itemize}
\item $\tau(v)[1] = \{A \in N \mid$ for all $v' \in \beta(v)$, $\sigma(v') \neq A\}$

\item $\tau(v)[2] = \{A \in N \mid$ for some $v' \in \beta(v)$, $\sigma(v') = A$, and for all marked $v' \in \beta(v)$, $\sigma(v') \neq A\}$

\item $\tau(v)[3] = \{A \in N \mid$ for some marked $v' \in \beta(v)$, $\sigma(v') = A\}$
\end{itemize}
Notice that for each $v$, $\tau(v)$ partitions $N$: every $A \in N$ occurs in exactly one of $\tau(v)[1]$, $\tau(v)[2]$, and $\tau(v)[3]$.  So there are $3^{|N|}$ possible values for $\tau(v)$.

\begin{restatable}{lemma}{limited}
  \label{limited}

Let $H = (v_0, \dotsc, v_m)$ be a path in a derivation tree $D$ from the root to a leaf (excluding the leaf) with more than $(|N|\cdot3^{|N|})^{|N|^2\cdot3^{|N|}+1}$ branch nodes.  Then there are $0 \leq b_1 < t_1 < t_2 \leq b_2 \leq m$ such that
\begin{itemize}

\item $\sigma(b_1) = \sigma(t_1)$ and $\sigma(t_2) = \sigma(b_2)$,
\item $b_2$ is in $\beta(b_1)$ and $t_2$ is in $\beta(t_1)$,
\item $\tau(b_1) = \tau(t_1)$, and
\item there is a branch node between $b_1$ and $t_1$ or between $t_2$ and $b_2$.


\end{itemize}\end{restatable}

\begin{proof}[Sketch]If $H$ is flat, i.e. if all of the nodes in $H$ have the same stack, then after $|N|\cdot3^{|N|}$ branch nodes, there will have been two nodes with the same $\sigma$ and $\tau$, with a branch node between them.  Then we can set $b_1$ and $t_1$ to these two nodes and set $t_2 = b_2 = m$, since $m$ will be in $\beta(v)$ for every node $v$ on the path, because $H$ is flat.  If $H$ is not flat, then consider just the ``base'' of $H$, i.e. the nodes in $H$ with the smallest stack.  These nodes are separated by ``hills'' in which the stack is bigger.  The base of $H$ can be viewed as a flat path with gaps corresponding to the hills.  Then at most $|N|\cdot3^{|N|}$ of the hills can contain branch nodes.  We can then use an inductive argument to bound the number of branch nodes in each hill.  In this argument, each hill is itself treated as a path, which is shorter than the original path $H$ and so subject to the induction.  Since node 0 and node $m$ in $H$ can serve as a potential $b_1$ and $b_2$ for any of the hills, each hill has fewer configurations of $\sigma$ and $\tau$ to ``choose from'' if it is to avoid containing nodes which could serve as $t_1$ and $t_2$.  Working out the details of the induction gives the bound stated in the lemma.\qed\end{proof}

We are now ready to state our pumping lemma for indexed languages, which generalizes the pumping lemma for ET0L languages of \cite{rabkin}.  As noted, this lemma allows arbitrary positions in a word to be designated as ``marked'', and then provides guarantees about the marked positions during the pumping operation.  The only difference between our pumping operation and that of Theorem 15 of \cite{rabkin} is that in the latter, there are guaranteed to be at least two marked positions in the $v_{\textbf{i},\textbf{j}}$ of part 4, whereas in our lemma, this $v_{\textbf{i},\textbf{j}}$ might not contain any marked positions and could even be an empty string (which nonetheless maps under $\phi$ to $u_\textbf{i}$, which does contain a marked position).

\begin{restatable}{theorem}{pumpable}
  \label{pumpable}

Let $L$ be an indexed language.  Then there is an $l \geq 0$ (which we will call a threshold for $L$) such that for any $w \in L$ with at least $l$ marked positions,

\begin{enumerate}
\item $w$ can be written as $w = u_1 u_2 \dotsm u_n$ and each $u_i$ can be written $u_i = v_{i,1} v_{i,2} \dotsm v_{i,n_i}$ (we will denote the set of subscripts of $v$, i.e. $\{(i,j) \mid 1 \leq i \leq n$ and $1 \leq j \leq n_i\}$, by $I$);
\item there is a map $\phi : I \rightarrow \{1,\dotsc,n\}$ such that if each $v_{i,j}$ is replaced with $u_{\phi(i,j)}$, then the resulting word is still in $L$, and this process can be applied iteratively to always yield a word in $L$;
\item if $v_{i,j}$ contains a marked position then so does $u_{\phi(i,j)}$;
\item there is an $(\textup{\textbf{i}},\textup{\textbf{j}}) \in I$ such that $\phi(\textup{\textbf{i}},\textup{\textbf{j}}) = \textup{\textbf{i}}$, and there is at least one marked position in $u_\textup{\textbf{i}}$ but outside of $v_{\textup{\textbf{i}},\textup{\textbf{j}}}$.
\end{enumerate}\end{restatable}

\begin{proof}[Sketch]We take a grounded indexed grammar $G$ with language $L$ and set the threshold $l$ using the bound from Lemma \ref{limited} together with some properties of the productions of $G$.  Then we take any $w \in L$ with at least $l$ marked positions and take a derivation tree $D$ for $w$.  Some path in $D$ from the root to a leaf then has enough branch nodes to give us the $b_1$, $t_1$, $t_2$, and $b_2$ from Lemma \ref{limited}.  We then need to construct the map $\phi$ and the factors $u_i$ and $v_{i,j}$.  To do this, we use the nodes in $\beta(b_1)$ and $\beta(t_1)$.  The nodes in $\beta(b_1)$ will correspond to $v_{i,j}$s and those in $\beta(t_1)$ will correspond to $u_i$s.  The operation $\phi$ will then map each node in $\beta(b_1)$ to a node in $\beta(t_1)$ with the same $\sigma$, and which is marked if the node being mapped is marked.  This is possible because $\tau(b_1) = \tau(t_1)$.  The justification for this construction is that in $D$, between $b_1$ and $t_1$ the stack grows from $x$ to $yx$ for some $x,y \in F^*$, and then shrinks back to $x$ between $\beta(t_1)$ and $\beta(b_1)$.  Since $\sigma(b_1) = \sigma(t_1)$, the steps between $b_1$ and $t_1$ can be repeated, growing the stack from $x$ to $yx$ to $yyx$ to $yyyx$, and so on.  Then the $y$s can be popped back off by repeating the steps between the nodes in $\beta(t_1)$ and $\beta(b_1)$.  This construction gives us parts 1, 2, and 3 of the theorem.  Part 4 follows from the fact that there is a branch node between $b_1$ and $t_1$ or between $t_2$ and $b_2$.  In the former case, the $u_\textup{\textbf{i}}$ in part 4 corresponds to the yield produced between $b_1$ and $t_1$ involving the branch node, and the $v_{\textup{\textbf{i}},\textup{\textbf{j}}}$ is specially constructed as an empty factor which maps to $u_\textup{\textbf{i}}$.  In the latter case, since $t_2$ is in $\beta(t_1)$, $b_2$ is in $\beta(b_1)$, and $\sigma(t_2) = \sigma(b_2)$, the $u_\textup{\textbf{i}}$ in part 4 corresponds to $t_2$ and the $v_{\textup{\textbf{i}},\textup{\textbf{j}}}$ corresponds to $b_2$.\qed\end{proof}

\vspace{-16pt}
\begin{figure}
\begin{center}
\includegraphics[width=0.7\textwidth]{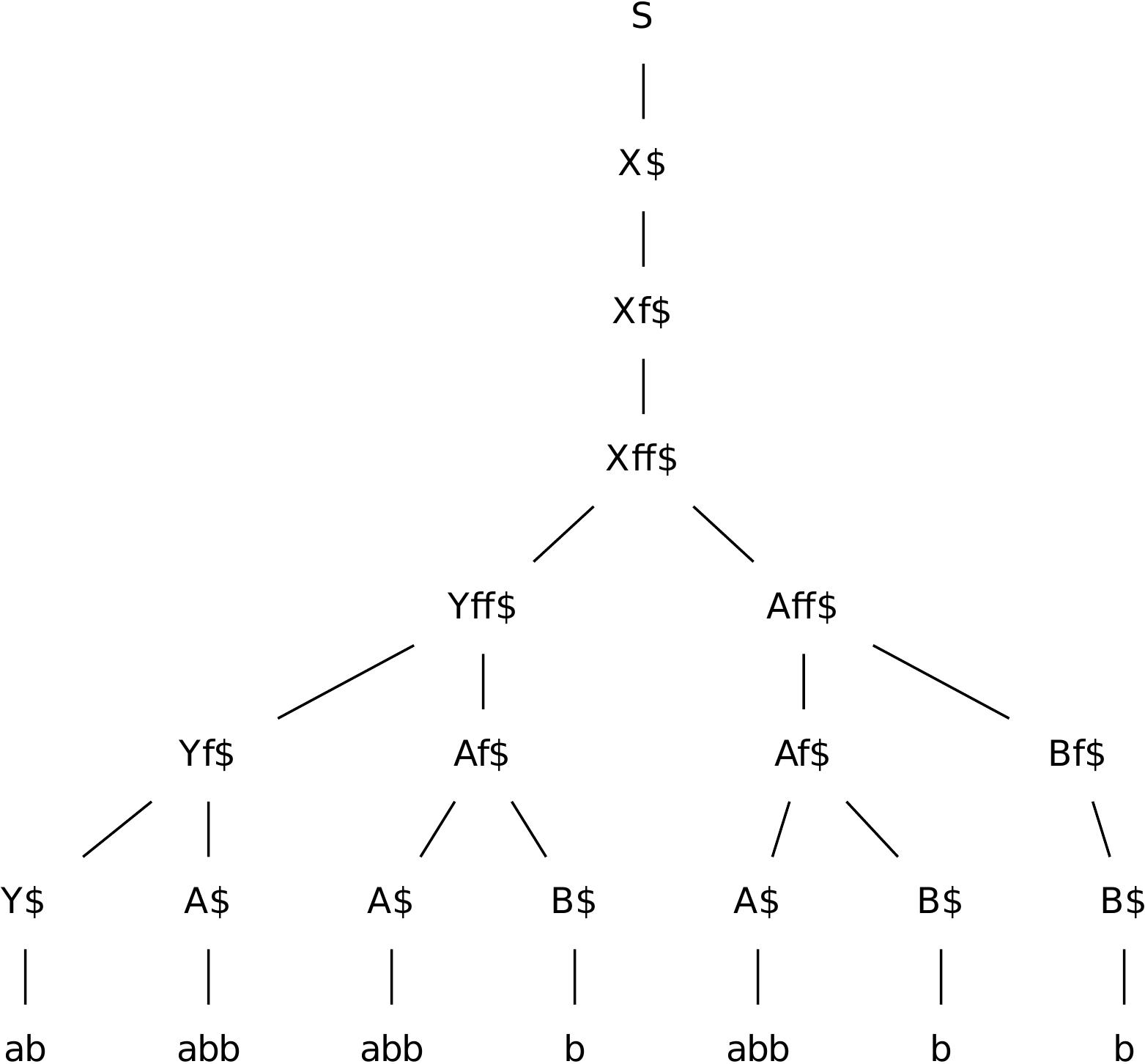}
\caption{A derivation tree for the string $ab^1ab^2ab^3ab^4$.}\label{fig:tree}
\end{center}
\end{figure}
\vspace{-26pt}

\begin{example}\label{example}
We now give an example of how our pumping operation works on a derivation tree of an indexed grammar.  Let the nonterminal alphabet $N$ be $\{S,X,Y,A,B\}$, the terminal alphabet $T$ be $\{a,b\}$, the stack alphabet $F$ be $\{f\}$, and the set of productions $P$ be $\{S \rightarrow X\$$, $X \rightarrow Xf$, $X \rightarrow YA$, $Yf \rightarrow YA$, $Y\$ \rightarrow ab$, $Af \rightarrow AB$, $A\$ \rightarrow abb$, $Bf \rightarrow B$, $B\$ \rightarrow b\}$.  Let $G$ be the indexed grammar $(N,T,F,P,S)$.  Notice that $G$ is grounded and that $L(G)$ determines the infinite word $ab^1ab^2ab^3\dotsm$.  Figure \ref{fig:tree} depicts a derivation tree of $G$ with terminal yield $ab^1ab^2ab^3ab^4$.  Notice how stacks are copied as the derivation proceeds; for example, the production $X \rightarrow YA$ applied to $Xff\$$ copies $X$'s stack $ff\$$ to both $Y$ and $A$, yielding $Yff\$\ Aff\$$.

Take every position in the terminal yield of the tree to be marked.  Let $b_1$ and $t_1$ be the nodes labelled $X\$$ and $Xf\$$, respectively, and let $t_2$ and $b_2$ be the nodes labelled $Af\$$ and $A\$$, respectively, in the $Yff\$$ subtree.  We have $\sigma(b_1) = \sigma(t_1) = X$, $\sigma(t_2) = \sigma(b_2) = A$, and $\tau(b_1) = \tau(t_1) = [\{S,X\}, \{\}, \{Y,A,B\}]$.  Additionally, $b_2$ is in $\beta(b_1)$, which consists of all the nodes labelled $A\$$, $B\$$, or $Y\$$, and $t_2$ is in $\beta(t_1)$, which consists of all the nodes labelled $Af\$$, $Bf\$$, or $Yf\$$.  Also, there is a branch node between $t_2$ and $b_2$, namely $t_2$ itself.  This satisfies the conditions of Lemma \ref{limited}.  We break up the string as $\overbracket[0.5pt][1pt]{\hspace{2pt}\overbracket[0.5pt][1pt]{ab} \overbracket[0.5pt][1pt]{abb}\hspace{2pt}} \overbracket[0.5pt][1pt]{\hspace{2pt}\overbracket[0.5pt][1pt]{abb} \overbracket[0.5pt][1pt]{b}\hspace{2pt}} \overbracket[0.5pt][1pt]{\hspace{2pt}\overbracket[0.5pt][1pt]{abb} \overbracket[0.5pt][1pt]{b}\hspace{2pt}} \overbracket[0.5pt][1pt]{\hspace{2pt}\overbracket[0.5pt][1pt]{b}\hspace{2pt}}$, where the outer brackets delimit the $u_i$s and the inner brackets delimit the $v_{i,j}$s.  Set $\phi(1,1) = 1$, $\phi(1,2) = 2$, $\phi(2,1) = 2$, $\phi(2,2) = 4$, $\phi(3,1) = 2$, $\phi(3,2) = 4$, and $\phi(4,1) = 4$.  Applying the pumping operation yields $\overbracket[0.5pt][1pt]{\hspace{2pt} \overbracket[0.5pt][1pt]{ab} \overbracket[0.5pt][1pt]{abb}\hspace{2pt}} \overbracket[0.5pt][1pt]{\hspace{2pt} \overbracket[0.5pt][1pt]{abb} \overbracket[0.5pt][1pt]{b}\hspace{2pt}} \overbracket[0.5pt][1pt]{\hspace{2pt} \overbracket[0.5pt][1pt]{abb} \overbracket[0.5pt][1pt]{b}\hspace{2pt}} \overbracket[0.5pt][1pt]{\hspace{2pt} \overbracket[0.5pt][1pt]{b}\hspace{2pt}} \overbracket[0.5pt][1pt]{\hspace{2pt} \overbracket[0.5pt][1pt]{abb} \overbracket[0.5pt][1pt]{b}\hspace{2pt}} \overbracket[0.5pt][1pt]{\hspace{2pt} \overbracket[0.5pt][1pt]{b}\hspace{2pt}} \overbracket[0.5pt][1pt]{\hspace{2pt} \overbracket[0.5pt][1pt]{b}\hspace{2pt}} = ab^1ab^2ab^3ab^4ab^5$, which is indeed in $L(G)$.  Applying it again yields $ab^1ab^2ab^3ab^4ab^5ab^6$, and so on.



\end{example}

We now follow \cite{rabkin} in giving a more formal description of the replacement operation in part 2 of Theorem \ref{pumpable}.  This operation produces the words $w^{(t)}$ for all $t \geq 0$, where
\vspace{-9pt}
\begin{align*}
v_{i,j}^{(0)} &= v_{i,j} \\
u_i^{(t)} &= v_{i,1}^{(t)} v_{i,2}^{(t)} \dotsm v_{i,n_i}^{(t)} \\
v_{i,j}^{(t+1)} &= u_{\phi(i,j)}^{(t)} \\
w^{(t)} &= u_1^{(t)} u_2^{(t)} \dotsm u_n^{(t)}
\end{align*}
Notice that $w^{(0)} = w$.  The following lemma states that the number of marked symbols tends to infinity as the replacement operation is repeatedly applied.

\begin{lemma}\label{infinite}If $L$ is an indexed language with threshold $l$, and $w \in L$ has at least $l$ marked symbols, then for all $t \geq 0$, $w^{(t)}$ has at least $t$ marked symbols.\end{lemma}

\begin{proof}
 
Call each $v_{i,j}$ a $v$-word and call a $v$-word marked if it contains a marked position.  We will show by induction on $t$ that for all $t \geq 0$, $w^{(t)}$ contains at least $t$ occurrences of marked $v$-words.  Obviously the statement holds for $t = 0$.  So say $t \geq 1$ and suppose for induction that $w^{(t-1)}$ contains at least $t-1$ occurrences of marked $v$-words.  By part 3 of Theorem \ref{pumpable}, for every marked $v_{i,j}$ in $w^{(t-1)}$, $v_{i,j}^{(1)}$ contains a marked position.  Then $w^{(t)}$ contains at least $t-1$ occurrences of marked v-words.  Now by part 4 of the theorem, there is an $(i,j) \in I$ such that $\phi(i,j) = i$ and there is at least one marked position in $u_i$ but outside of $v_{i,j}$.  Then $v_{i,j}^{(1)} = u_i$ contains at least one more occurrence of a marked $v$-word than $v_{i,j}$.  Now since $w$ contains $v_{i,j}$, $w^{(t-1)}$ contains $v_{i,j}$.  Then $w^{(t)}$ contains at least one more occurrence of a marked $v$-word than $w^{(t-1)}$.  So $w^{(t)}$ contains at least $t$ occurrences of marked $v$-words, completing the induction.  Thus for all $t \geq 0$, $w^{(t)}$ contains at least $t$ occurrences of marked $v$-words, hence $w^{(t)}$ contains at least $t$ marked positions.\qed

\end{proof}

\section{Applications}\label{sec:applications}

In this section we give some applications of our pumping lemma for indexed languages (Theorem \ref{pumpable}).  We prove for IL a theorem about frequent and rare symbols which is proved in \cite{rabkin} for ET0L languages.  Then we characterize the infinite words determined by indexed languages.

\subsection{Frequent and rare symbols}

Let $L$ be a language over an alphabet $A$, and $B \subseteq A$.  $B$ is \textbf{nonfrequent} if there is a constant $c_B$ such that $\#_B(w) \leq c_B$ for all $w \in L$.  Otherwise it is called \textbf{frequent}.  $B$ is called \textbf{rare} if for every $k \geq 1$, there is an $n_k \geq 1$ such that for all $w \in L$, if $\#_B(w) \geq n_k$ then the distance between any two appearances in $w$ of symbols from $B$ is at least $k$.

\begin{theorem}\label{rare}Let $L$ be an indexed language over an alphabet $A$, and $B \subseteq A$.  If $B$ is rare in $L$, then $B$ is nonfrequent in $L$.\end{theorem}

\begin{proof}

Suppose $B$ is rare and frequent in $L$.  By Theorem \ref{pumpable}, $L$ has a threshold $l \geq 0$.  Since $B$ is frequent in $L$, there is a $w \in L$ with more than $l$ symbols from $B$.  If we mark them all, parts 1 to 4 of the theorem apply.  By part 3 of the theorem, if $v_{i,j}$ contains a marked position then so does $u_{\phi(i,j)}$, and by part 4 of the theorem, there is an $(i,j) \in I$ such that $\phi(i,j) = i$ and there is at least one marked position in $u_i$ but outside of $v_{i,j}$.  Then $u_i^{(1)}$ contains at least two marked positions.  So take any two marked positions in $u_i^{(1)}$ and let $d$ be the distance between them.  Let $k = d+1$.  Since $B$ is rare in $L$, there is an $n_k \geq 1$ such that for all $w' \in L$, if $\#_B(w') \geq n_k$ then the distance between any two appearances in $w'$ of symbols from $B$ is at least $k$.  By Lemma \ref{infinite}, $w^{(n_k)}$ contains at least $n_k$ marked symbols, so $\#_B(w^{(n_k)}) \geq n_k$.  Then the distance between any two appearances in $w^{(n_k)}$ of symbols from $B$ is at least $k$.  Since $u_i$ appears in $w$ and $u_i^{(1)}$ contains $u_i$, $u_i^{(1)}$ appears in $w^{(t)}$ for all $t \geq 1$.  Then $u_i^{(1)}$ appears in $w^{(n_k)}$ and contains two symbols from $B$ separated by $d < k$, a contradiction.  So if $B$ is rare in $L$, then $B$ is nonfrequent in $L$.\qed

\end{proof}

Theorem \ref{rare} gives us an alternative proof of the result of Hayashi \cite{hayashi} and Gilman \cite{gilman} that the language $L$ below is not indexed.

\begin{corollary}[\cite{hayashi} Theorem 5.3; \cite{gilman} Corollary 4]The language $L = \{(\textup{\texttt{ab}}^n)^n \mid n \geq 1\}$ is not indexed.\end{corollary}

\begin{proof}The subset $\{\texttt{a}\}$ of $\{\texttt{a},\texttt{b}\}$ is rare and frequent in $L$.  So by Theorem \ref{rare}, $L$ is not indexed.\qed\end{proof}

\subsection{CD0L and morphic words}

Next, we turn to characterizing the infinite words determined by indexed languages.  We show that every infinite indexed language has an infinite CD0L subset, which then implies that $\omega(IL)$ contains exactly the morphic words.  This is a new result which we were not able to obtain using the pumping lemmas of \cite{hayashi} or \cite{gilman}.

\begin{theorem}\label{indexcdol}Let $L$ be an infinite indexed language.  Then $L$ has an infinite CD0L subset.\end{theorem}

\begin{proof}

By Theorem \ref{pumpable}, $L$ has a threshold $l \geq 0$.  Take any $w \in L$ such that $|w| \geq l$, and mark every position in $w$.  Then parts 1 to 4 of the theorem apply.  Now for each $(i,j) \in I$, create a new symbol $x_{i,j}$.  Let $X$ be the set of these symbols.  For $i$ from 1 to $n$, let $x_i = x_{i,1} x_{i,2} \dotsm x_{i,n_i}$.  Let $x = x_1 x_2 \dotsm x_n$.  Let $h$ be a morphism such that $h(x_{i,j}) = x_{\phi(i,j)}$ for all $(i,j) \in I$.  Let $g$ be a morphism such that $g(x_{i,j}) = v_{i,j}$ for all $(i,j) \in I$.  Let $A$ be the alphabet of $L$.  Let $G$ be the HD0L system $(X \cup A,h,x,g)$.  Then for all $t \geq 0$, $g(h^t(x)) = w^{(t)}$.  By Lemma \ref{infinite}, for all $t \geq 0, |w^{(t)}| \geq t$.  Then $L(G)$ is an infinite HD0L subset of $L$.  By Theorem 18 of \cite{smith2}, every infinite HD0L language has an infinite CD0L subset.  Therefore $L$ has an infinite CD0L subset.\qed


\end{proof}

\begin{theorem}\label{morphic}$\omega$(IL) contains exactly the morphic words.\end{theorem}

\begin{proof}For any infinite word $\alpha \in \omega(IL)$, some $L \in IL$ determines $\alpha$.  Then $L$ is an infinite indexed language, so by Theorem \ref{indexcdol}, $L$ has an infinite CD0L subset $L'$.  Then $L'$ determines $\alpha$, so $\alpha$ is in $\omega$(CD0L).  Then by Theorem 23 of \cite{smith2}, $\alpha$ is morphic.  For the other direction, by Theorem 23 of \cite{smith2}, every morphic word is in $\omega$(CD0L), so since CD0L $\subset$ IL, every morphic word is in $\omega$(IL).\qed\end{proof}
 
Theorem \ref{morphic} lets us use existing results about morphic words to show that certain languages are not indexed, as the following example shows.

\begin{corollary}Let $L = \{\textup{\texttt{0}}$, $\textup{\texttt{0:1}}$, $\textup{\texttt{0:1:01}}$, $\textup{\texttt{0:1:01:11}}$, $\textup{\texttt{0:1:01:11:001}}$, $\dotsc\}$, the language containing for each $n \geq 0$ a word with the natural numbers up to $n$ written in backwards binary and colon-separated.  Then $L$ is not indexed.\end{corollary}

\begin{proof}$L$ determines the infinite word $\alpha = \texttt{0:1:01:11:001:101:011:111:}\dotsm$.  By Theorem 3 of \cite{ck}, $\alpha$ is not morphic.  Then by our Theorem \ref{morphic}, $\alpha$ is not in $\omega$(IL), so no language in IL determines $\alpha$, hence $L$ is not indexed.\qed\end{proof}

\section{Conclusion}\label{sec:conclusion}

In this paper we have characterized the infinite words determined by indexed languages, showing that they are exactly the morphic words.  In doing so, we proved a new pumping lemma for the indexed languages, which may be of independent interest and which we hope will have further applications.  One direction for future work is to look for more connections between formal languages and infinite words via the notion of prefix languages.  It would be interesting to see what other language classes determine the morphic words, and what language classes are required to determine infinite words that are not morphic.  More generally, for any language class, we can ask what class of infinite words it determines, and for any infinite word, we can ask in what language classes it can be determined, yielding many opportunities for future research.  It is hoped that work in this area will help to build up a theory of the complexity of infinite words with respect to what language classes can determine them.

\subsubsection*{Acknowledgments.} I want to thank my advisor, Rajmohan Rajaraman, for supporting this work, encouraging me, and offering many helpful comments and suggestions.

\bibliographystyle{splncs03}
\bibliography{../sources}

\newpage

\appendix
\section{Appendix}

We give some proofs which were sketched or omitted from the body.

\subsection{Proposition \ref{grounded}}

\grounded*

\begin{proof}Let $G = (N,T,F,P,S)$.  Add a nonterminal $S'$ to $N$ and replace every occurrence of $S$ in $P$ with $S'$.  Then add a symbol $\$$ to $F$ and add to $P$ the production $S \rightarrow S'\$$.  Next, for every $t \in T$, add a nonterminal $X_t$ to $N$ and replace every occurrence of $t$ in $P$ with $X_t$.  Then add a nonterminal $X_\lambda$ to $N$ and replace every production of the form $A \rightarrow \lambda$ with $A \rightarrow X_\lambda$ and every production of the form $Af \rightarrow \lambda$ with $Af \rightarrow X_\lambda$.  Finally, for every $s \in T \cup \{\lambda\}$, add to $P$ the production $X_s\$ \rightarrow s$ and for every $f \in F$, add to $P$ the production $X_s f \rightarrow X_s$.  The resulting grammar $G'$ is grounded and $L(G') = L(G)$.\qed\end{proof}

\subsection{Lemma \ref{limited}}

We now give a full proof of Lemma \ref{limited}, filling out the sketch given in the body.  We prove two supporting lemmas in order to prove the main lemma.  First we give some definitions.  Let $G = (N,T,F,P,S)$ be a grounded indexed grammar and let $D$ be a derivation tree of $G$.  For a node $v$, $D(v)$ means the subtree of $D$ whose root is $v$.  We denote the terminal yield of $D(v)$ by yield($v$).  For nodes $v_1,v_2$ in $D$, we say that $v_1$ is to the left of $v_2$, and $v_2$ is to the right of $v_1$, if neither node is descended from the other, and if $v_1$ would be encountered before $v_2$ in a depth-first traversal of $D$ in which edges are chosen from left to right.  A node $v_1$ is \textbf{reachable} from a node $v_2$ if $v_1$ is identical to $v_2$ or descended from $v_2$.

We will be working with a variation of a path which we call a descent.  A \textbf{descent} $H$ in $D$ is a list of internal nodes $(v_0, \dotsc, v_m)$ with $m \geq 0$ such that $v_m$ is in $\beta(v_0)$ and for each $i \geq 1$, $v_i$ is a descendant (not necessarily a child) of $v_{i-1}$.  As with paths, we will sometimes refer to nodes in $H$ by their indices; e.g. $\sigma(i)$ in the context of $H$ means $\sigma(v_i)$.  For $0 \leq i < m$, we say there is a \textbf{split} in $H$ between $i$ and $i+1$ iff any of the nodes on the path in $D$ from $v_i$ (inclusive) to $v_{i+1}$ (exclusive) is a branch node.  If this path in $D$ has more than one branch node, we still say that there is just one split between $i$ and $i+1$ in $H$.  Thus $H$ has at most $m$ splits.


We call $H$ \textbf{controlled} if for all $0 < i \leq m$, $i$ is a child (not merely a descendant) of $i-1$.  Notice that any path in $D$ from the root to a leaf (excluding the leaf) is a controlled descent.  We call $H$ \textbf{flat} if for all $0 \leq i \leq m$, $\eta(i) = \eta(0)$.  We call $H$ \textbf{limited} iff for all $0 \leq b_1 < t_1 < t_2 \leq b_2 \leq m$ such that
\begin{itemize}

\item $\sigma(b_1) = \sigma(t_1)$ and $\sigma(t_2) = \sigma(b_2)$,
\item $b_2$ is in $\beta(b_1)$ and $t_2$ is in $\beta(t_1)$, and
\item $\tau(b_1) = \tau(t_1)$,

\end{itemize}
there are no splits between $b_1$ and $t_1$ or between $t_2$ and $b_2$.

\begin{lemma}\label{split2}Let $H = (v_0, \dotsc, v_m)$ be a limited flat descent.  Then $H$ has at most $|N|\cdot3^{|N|}$ splits.\end{lemma}

\begin{proof}

For any $i$ in $H$, there are $|N|$ possible values for $\sigma(i)$ and $3^{|N|}$ possible values for $\tau(i)$.  Suppose there are more than $|N|\cdot3^{|N|}$ splits between $0$ and $m$.  Then there are $b_1,t_1$ such that $0 \leq b_1 < t_1 < m$, $\sigma(b_1) = \sigma(t_1)$, $\tau(b_1) = \tau(t_1)$, and there is a split between $b_1$ and $t_1$.  Let $t_2 = b_2 = m$.  Obviously $\sigma(t_2) = \sigma(b_2)$.  By the definition of a descent, $m$ is in $\beta(0)$.  Then since $H$ is flat, $m$ is in $\beta(i)$ for all $i$ in $H$.  Hence $b_2$ is in $\beta(b_1)$ and $t_2$ is in $\beta(t_1)$.  But then $H$ is not limited, a contradiction.  So $H$ has at most $|N|\cdot3^{|N|}$ splits.\qed

\end{proof}

Let $R$ be the set of possible values of $\tau(v)$; i.e. the set of 3-tuples each of which partitions $N$.  We have $|R| = 3^{|N|}$.  For any descent $H = (v_0, \dotsc, v_m)$ and $W \subseteq N \times N \times R$, $H$ \textbf{respects} $W$ iff for every $0 \leq i < j \leq m$ such that $j$ is in $\beta(i)$, the triple $(\sigma(i), \sigma(j), \tau(i))$ is in $W$.


\begin{lemma}\label{split3}Take any $W \subseteq N \times N \times R$.  Any limited controlled descent which respects $W$ has at most $(|N|\cdot3^{|N|})^{|W|+1}$ splits.\end{lemma}

\begin{proof}

Let $k = |N|\cdot3^{|N|}$, the bound from Lemma \ref{split2}.  For $x \in \mathbb{N}$, let $f(x) = k^{x+1}-2k$.  We will show by induction on $|W|$ that any limited controlled descent which respects $W$ has at most $f(|W|)$ splits.

If $|W| = 0$, then no limited controlled descent respects $W$, so the statement holds trivially.  So say $|W| \geq 1$.  Suppose for induction that for every $W'$ such that $|W'| < |W|$, any limited controlled descent which respects $W'$ has at most $f(|W'|)$ splits.

Take any $n$ such that there is a limited controlled descent which respects $W$ and has exactly $n$ splits.  We will show that $n \leq f(|W|)$.  Take the lowest $m$ such that there is a limited controlled descent $H = (v_0, \dotsc, v_m)$ which respects $W$ and has exactly $n$ splits.

Suppose $m = 0$.  Then $n = 0$.  Since $|N| \geq 1$, $k \geq 3$.  Then since $|W| \geq 1$ and $f$ is increasing, $f(|W|) \geq f(1) \geq 3^{1+1}-2\cdot3 \geq 3$.  Then $n \leq f(|W|)$ as desired.  So say $m \geq 1$.

Let $base$ be the list consisting of every node $v_i$ in $H$ for which $\eta(i) = \eta(0)$.  Call $i,j$ base-adjacent iff $i < j$, $i$ and $j$ are in $base$, and no node between $i$ and $j$ is in $base$.  Call the interval from $i$ to $j$ a hill iff $i,j$ are base-adjacent.  Call a hill a split hill iff it contains at least one split.

Take any base-adjacent $i,j$ such that the hill between $i$ and $j$ has at least as many splits as any hill between 0 and $m$.  Let $x$ be the number of splits between $i$ and $j$.  We will show $x \leq 2 + f(|W|-1)$.  Suppose $j = i+1$.  Then $x \leq 1$.  Suppose $j = i+2$.  Then $x \leq 2$.  So say $j > i+2$.  Let $i' = i+1$ and $j' = j-1$.  Then $0 < i' < j' < m$, and since $H$ is controlled, $\eta(i') = \eta(j') = 1 + \eta(0)$ and $j'$ is in $\beta(i')$.

Suppose there are no splits between 0 and $i'$ or between $j'$ and $m$.  Then there are $n$ splits between $i'$ and $j'$.  Let $m' = j'-i'$.  Then the limited controlled descent $(v_{i'}, \dotsc, v_{j'})$ respects $W$ and has $n$ splits.  But $m' < m$, a contradiction, since by the construction of $m$, there is no such $m'$.

So there is such a split.  Now, since $m$ is in $\beta(0)$, $(\sigma(0),\sigma(m),\tau(0))$ is in $W$.  Let $W'$ = $W - (\sigma(0),\sigma(m),\tau(0))$.  Suppose there are $i'',j''$ such that $i' \leq i'' < j'' \leq j'$, $j''$ is in $\beta(i'')$, $\sigma(i'') = \sigma(0)$, $\sigma(j'') = \sigma(m)$, and $\tau(i'') = \tau(0)$.  Let $b_1=0$, $t_1=i''$, $t_2=j''$, and $b_2=m$.  Since there is a split between 0 and $i'$ or between $j'$ and $m$, there is a split between 0 and $i''$ or between $j''$ and $m$, hence between $b_1$ and $t_1$ or between $t_2$ and $b_2$.  But then $H$ is not limited, a contradiction.  So there are no such $i'',j''$.  Hence the limited controlled descent $(v_{i'}, \dotsc, v_{j'})$ respects $W'$.  Since $|W'| < |W|$, by the induction hypothesis there are at most $f(|W'|)$ splits between $i'$ and $j'$.  Then allowing a split between $i$ and $i'$ and a split between $j'$ and $j$, $x \leq 2 + f(|W'|)$.

Now, the list $base$ is a flat descent.  If $base$ was not limited, then $H$ would not be limited.  So $base$ is limited.  Then by Lemma \ref{split2}, it has at most $k$ splits.  Then there are at most $k$ split hills in $H$.  Recall that $k \geq 3$.  Each split hill has at most $x$ splits.  Therefore
\begin{align*}
n &\leq kx \\
n &\leq k(2 + f(|W|-1)) \\
n &\leq 2k + k(k^{|W|-1+1}-2k) \\
n &\leq 2k + k^{|W|+1} - 2k^2 \\
n &\leq f(|W|),
\end{align*}
completing the induction.  Therefore any limited controlled descent which respects $W$ has at most $(|N|\cdot3^{|N|})^{|W|+1}$ splits.\qed\end{proof}

\limited*

\begin{proof}Since $H$ is a path from the root to a leaf, $H$ is a controlled descent.  Suppose $H$ is limited.  Let $W = N \times N \times R$.  Clearly $H$ respects $W$.  Then by Lemma \ref{split3}, $H$ has at most $(|N|\cdot3^{|N|})^{|W|+1}$ = $(|N|\cdot3^{|N|})^{|N|^2\cdot3^{|N|}+1}$ splits.  Then since by definition, every branch node is immediately followed by a split, $H$ has at most $(|N|\cdot3^{|N|})^{|N|^2\cdot3^{|N|}+1}$ branch nodes, a contradiction.  So $H$ is not limited and the lemma holds.\qed\end{proof}

\subsection{Theorem \ref{pumpable}}

Finally, we give a full proof of Theorem \ref{pumpable}, filling out the sketch given in the body.

\pumpable*

\begin{proof}Let $G = (N,T,F,P,S)$ be a grounded indexed grammar such that $L(G) = L$.  Let $d$ be the highest $i$ such that there is a production in $P$ with $i$ nonterminals on the righthand side.  Let $e$ be the highest $i$ such that there is a production in $P$ with $i$ terminals on the righthand side.  Let $z = (|N|\cdot3^{|N|})^{|N|^2\cdot3^{|N|}+1}$.  Let $l = ed^z+1$.

If $L$ is finite, then trivially the theorem holds.  So say $L$ is infinite.  Let $w$ be a word in $L$ with at least $l$ marked positions.  Take any derivation tree $D$ of $w$.  Suppose no path in $D$ from the root to a leaf has more than $z$ branch nodes.  The maximum outdegree of $D$ is at most $d$.  Then by Lemma 14 of \cite{rabkin}, $D$ has at most $d^z$ marked leaves.  Recall that each leaf has a label in $T^*$ (a terminal string) and marked leaves are those whose label contains a marked position of $w$.  Each marked leaf has a label of length at most $e$.  But then $w$ has at most $ed^z$ marked positions, a contradiction.

So some path $H = (v_0, \dotsc, v_m)$ in $D$ from the root to a leaf (excluding the leaf) has more than $z$ branch nodes.  Then by Lemma \ref{limited}, there are $0 \leq b_1 < t_1 < t_2 \leq b_2 \leq m$ such that $\sigma(b_1) = \sigma(t_1)$, $\sigma(t_2) = \sigma(b_2)$, $b_2$ is in $\beta(b_1)$, $t_2$ is in $\beta(t_1)$, $\tau(b_1) = \tau(t_1)$, and there is a branch node between $b_1$ and $t_1$ or between $t_2$ and $b_2$.

We specify the subscripts of $I$ by defining $n$ and each $n_i$.  Let $n = 4 + |\beta(t_1)|$.  Let $n_1 = n_n = 1$.  We will give names to the nodes in $\beta(b_1)$ and $\beta(t_1)$, as follows.  Let $n_2$ be the number of nodes in $\beta(b_1)$ to the left of $t_1$, plus one.  From left to right, call these nodes $N_{2,2}, N_{2,3}, \dotsc, N_{2,n_2}$.  Call the nodes in $\beta(t_1)$ from left to right $N_3, N_4, \dotsc, N_{n-2}$.  For each $i$ from 3 to $n-2$, let $n_i$ be the number of nodes in $\beta(b_1)$ which are reachable from $N_i$.  From left to right, call these nodes $N_{i,1}, N_{i,2}, \dotsc, N_{i,n_i}$.  Let $n_{n-1}$ be the number of nodes in $\beta(b_1)$ to the right of $t_1$, plus one.  From left to right, call these nodes $N_{n-1,1}, N_{n-1,2}, \dotsc, N_{n-1,n_{n-1}-1}$.

We now define each $u_i$ and each $v_{i,j}$, as well as the map $\phi$ which will be used in the replacement operation.  Let $v_{1,1}$ be the yield of $D$ to the left of $b_1$, and let $\phi(1,1) = 1$.  Let $v_{2,1} = \lambda$ and let $\phi(2,1) = 2$.  Let $v_{n-1,n_{n-1}} = \lambda$ and let $\phi(n-1,n_{n-1}) = n-1$.  Let $v_{n,1}$ be the yield of $D$ to the right of $b_1$, and let $\phi(n,1) = n$.  For all other $i,j$ for which there is a node $N_{i,j}$, proceed as follows.  Set $v_{i,j}$ = yield$(N_{i,j})$.  If $N_{i,j}$ is $b_2$, then set $\phi(i,j) = i$.  (Notice that $N_i = t_2$, so $\sigma(N_i) = \sigma(N_{i,j})$.)  Otherwise, if yield($N_{i,j}$) contains a marked position, then $\sigma(N_{i,j})$ is in $\tau(b_1)[3]$.  Since $\tau(b_1) = \tau(t_1)$, there is an $N_k$ such that $\sigma(N_k) = \sigma(N_{i,j})$ and yield($N_k$) contains a marked position.  Set $\phi(i,j) = k$.  Otherwise, yield($N_{i,j}$) does not contain a marked position, so $\sigma(N_{i,j})$ is in $\tau(b_1)[2]$.  Then there is an $N_k$ such that $\sigma(N_k) = \sigma(N_{i,j})$.  Set $\phi(i,j) = k$.  Finally, for $i$ from 1 to $n$, let $u_i = v_{i,1} \dotsm v_{i,n_i}$.

We now have that each $u_i$ is the yield of the nodes in some part of $D$.  In particular, $u_1$ is the yield to the left of $b_1$, $u_2$ is the yield under $b_1$ to the left of $t_1$, $u_3 \dotsm u_{n-2}$ is the yield under $t_1$, $u_{n-1}$ is the yield under $b_1$ to the right of $t_1$, and $u_n$ is the yield to the right of $b_1$.  Thus $u_1 \dotsm u_n$ is the yield of $D$, namely $w$.  This gives us part 1 of the theorem.

To establish part 2, we will argue that the derivation tree $D$ of $s$ can be ``pumped'' to produce new derivation trees which yield strings in accordance with the replacement operation.  Let $x \in F^*$ be the stack at node $b_1$.  Since $b_2$ is in $\beta(b_1)$, and $b_2$ is a descendant of $t_1$, the stack at node $t_1$ has the form $yx$ for some $y \in F^*$.  Thus in $D$, the stack grows from $x$ to $yx$ between $b_1$ and $t_1$, remains at or above $yx$ until $t_2$, and then shrinks back to $x$ at $b_2$.  We will construct a new derivation tree $D'$ in which the stack grows from $x$ to $yx$ to $yyx$, then shrinks from $yyx$ to $yx$ to $x$.  The construction is as follows.  Initialize $D'$ to a copy of $D$.  Next, make a copy $C$ of the subtree $D(b_1)$.  In $C$, for every ancestor of any node in $\beta(b_1)$, put $y$ on top of its stack.  For each node $N_{i,j}$ in $\beta(b_1)$, proceed as follows.  Notice that $N_{i,j}$ has stack $x$ in $D$ and hence stack $yx$ in $C$, while $N_{\phi(i,j)}$ has stack $yx$ in $D$.  Further, $\sigma(N_{i,j}) = \sigma$($N_{\phi(i,j)}$).  So in $C$, replace the subtree $C(N_{i,j})$ with the subtree $D$($N_{\phi(i,j)}$).  Finally, in $D'$, replace the subtree $D'(t_1)$ with $C$.  The resulting derivation tree $D'$ now obeys the rules of $G$ and has a yield equal to the result of performing the replacement operation of part 2 on $w$.  This procedure can be repeated to produce a new tree $D''$ in which the stack grows to $yyyx$, with a yield corresponding to two iterations of the replacement operation of part 2, and so on.  This gives us part 2 of the theorem.  Part 3 follows from the construction of the $\phi$ operation.

We now establish part 4 of the theorem.  Recall that there is a branch node in $H$ between $b_1$ and $t_1$ or between $t_2$ and $b_2$.  If there is a branch node between $b_1$ and $t_1$, then by definition, this branch node has at least two marked children.  Take any one of these marked children which is not on the path from node $b_1$ to node $t_1$.  This child is either to the left of $t_1$ or to the right of $t_1$.  If it is to the left of $t_1$, then some marked $N_{2,i}$ is reachable from it for some $2 \leq i \leq n_2$.  Since $N_{2,i}$ is marked, yield($N_{2,i}$) contains a marked position.  So $v_{2,i}$ contains a marked position.  Then since $\phi(2,1) = 2$ and $u_2$ contains $v_{2,1}$ and $v_{2,i}$, (2,1) satisfies part 4.  Similarly, if the marked child is to the right of $t_1$, then some marked $N_{n-1,i}$ is reachable from it for some $1 \leq i \leq n_{n-1}-1$.  Since $N_{n-1,i}$ is marked, yield($N_{n-1,i}$) contains a marked position.  So $v_{n-1,i}$ contains a marked position.  Then since $\phi(n-1,n_{n-1}) = n-1$ and $u_{n-1}$ contains $v_{n-1,i}$ and $v_{n-1,n-1}$, $(n-1,n-1)$ satisfies part 4.  Otherwise, if there is no branch node in $H$ between $b_1$ and $t_1$, then there is a branch node between $t_2$ and $b_2$.  We have $t_2 = N_i$ and $b_2 = N_{i,j}$ for some $3 \leq i \leq n-2$ and $1 \leq j \leq n_i$.  By definition, the branch node between $t_2$ and $b_2$ has at least two marked children.  Take any one of these marked children which is not on the path from node $t_2$ to node $b_2$.  Some marked $N_{i,k}$ is reachable from this child where $k \neq j$.  So yield($N_{i,k}$) contains a marked position, hence $v_{i,k}$ contains a marked position.  Since $N_{i,j}$ is $b_2$, $\phi(i,j) = i$ by construction.  Then since $u_i$ contains both $v_{i,j}$ and $v_{i,k}$, $(i,j)$ satisifies part 4.  This establishes part 4 of the theorem, which completes the proof.\qed

\end{proof}

\end{document}